\documentclass[12pt]{amsart}
\usepackage{amssymb,amsmath}

\usepackage{t1enc}
\usepackage[latin1]{inputenc}
\usepackage[german,english]{babel}
\usepackage{amsfonts}
\usepackage[all]{xy}
\usepackage{graphicx}
\usepackage{color}
\usepackage{tikz}
\usetikzlibrary{arrows,decorations.pathmorphing,backgrounds,positioning,fit,petri}

\usepackage{geometry}
\newtheorem{defi}{Definition}[section]
\newtheorem{prop}[defi]{Proposition}

\newtheorem{coro}[defi]{Corollary}

\theoremstyle{remark}
\newtheorem{rem}[defi]{Remark}
\newtheorem{exa}[defi]{Example}

\begin{document}

\title[Key establishment for left distributive systems]{Non-associative key establishment for left distributive systems}
\author{Arkadius Kalka and Mina Teicher}
\address{Department of Mathematics,
Bar Ilan University,
Ramat Gan 52900,
Israel}
\email{Arkadius.Kalka@rub.de, teicher@math.biu.ac.il}
\urladdr{http://homepage.ruhr-uni-bochum.de/arkadius.kalka/}

\begin{abstract} 
We construct non-associative key establishment protocols for all left self-distributive (LD), multi-LD-, and other left distributive systems.
Instantiations of these protocols using generalized shifted conjugacy in braid groups 
lead to instances of a natural and apparently new group-theoretic problem, which we call the (subgroup) conjugacy coset problem.
\end{abstract} 



\maketitle

\section{Introduction}
In an effort to construct new key establishment protocols (KEPs), which are hopefully harder to break than previously proposed non-commutative schemes, the first author introduced in his PhD thesis \cite{Ka07} (see also \cite{Ka12}) the first non-associative generalization
of the Anshel-Anshel-Goldfeld KEP \cite{AAG99}, which  revolutionized the field of \emph{non-commutative public key cryptography} (PKC) more than ten years ago.
For an introduction to non-commutative public key cryptography we refer to the book by Myasnikov et al. \cite{MSU11}. 
For further motivation and on {\it non-associative PKC} we refer to \cite{Ka12}. It turns out (see \cite{Ka12}) that in the context of AAG-like KEPs for magmas, left self-distributive systems
(LD-systems) and their generalizations (like multi-LD-systems) naturally occur. Though we constructed several examples of KEPs for non-associative 
LD- and multi-LD-systems \cite{Ka12}, we did not provide a general method to construct a KEP that works for all LD- and multi-LD-systems.
We fill this gap in the present paper. With this method at hand any LD- or multi-LD-system automatically provides a KEP - while in \cite{Ka07, Ka12}
we had to construct the key establishment functions for each example by hand. Therefore, we obtain a rich variety of new non-associatiave KEPs
coming from LD-, multi-LD-, and other left distributive systems.  
\par Instantiations of the proposed KEPs with concrete parameter values are left for future works \cite{KT13}. 

\medskip

\par {\bf Outline.} In section 2 we review LD-, multi-LD-, and other left distributive systems with many examples.  
Section 3 describes a KEP for all LD-systems with a discussion of related base problems. In section 4 we describe and analyze a KEP which 
does not only apply for all multi-LD-systems, but also for a big class of partial multi-LD-systems. 
Finally, in section 5 we discuss instantiations of these general protocols using generalized shifted conjugacy in braid groups. 
An associated base problem leads to an instance of a natural and apparently new group-theoretic problem, which we call the (subgroup) conjugacy coset problem.
 
\section{LD-systems and other distributive systems}

\subsection{Definitions}
\begin{defi} 
An {\it LD-system} $(S,*)$ is a set $S$ equipped with a binary operation $*$ on $S$ which satisfies the {\it left self-distributivity law} 
\[ x*(y*z)=(x*y)*(x*z) \quad {\rm for} \,\, {\rm all} \,\, x,y,z\in S. \] 
\end{defi}

\begin{defi} (Section X.3. in \cite{De00})
Let $I$ be an index set. A \emph{multi-LD-system}
 $(S,(*_i)_{i\in I})$ is a set $S$ equipped with a family of binary operations $(*_i)_{i\in I}$ on $S$ such that 
\[ x*_i(y*_jz)=(x*_iy)*_j(x*_iz) \quad {\rm for} \,\, {\rm all} \,\, x,y,z\in S \]
is satisfied for every $i,j$ in $I$. Especially, it holds for $i=j$, i.e., $(S,*_i)$ is an LD-system. If $|I|=2$ then we call $S$ a \emph{bi-LD-system}. 
\end{defi}

More vaguely, we will also use the terms \emph{partial multi-LD-system} and simply \emph{left distributive system} if the laws of a multi-LD-system
are only fulfilled for special subsets of $S$ or if only some of these (left) distributive laws are satisfied. 

\medskip

We begin with some examples of LD-systems taken from \cite{De06}. 

\medskip

{\bf 1.} We begin with a trivial example. $(S,*)$ with $x*y=f(y)$ is an LD-system for any function $f: S\rightarrow S$. 

\medskip

{\bf 2.} A set $S$ with a binary operation $*$, that satisfies no other relations than those resulting from the left self-distributivity law, is a free LD-system. Free LD-systems are studied extensively in \cite{De00}. 

\medskip

{\bf 3.} A classical example of an LD-system is $(G,*)$ where $G$ is a group equipped with the conjugacy operation $x*y=x^{-1}yx$ (or $x*^{\rm rev}y=xyx^{-1}$).
 Note that such an LD-system cannot be free, because conjugacy satisfies additionally the idempotency law $x*x=x$. 

\medskip

{\bf 4.} Finite groups equipped with the conjugacy operation are not the only finite LD-systems. Indeed, the socalled \emph{Laver tables} provide the classical example for finite LD-systems. 
There exists for each $n\in \mathbb{N}$ an unique LD-system $L_n=(\{ 1, 2, \ldots , 2^n \},*)$ with $k*1=k+1$.
The values for $k*l$ with $l\ne 1$ can be computed by induction using the left self-distributive law.
The Laver tables for $n=1,2,3$ are
\[ \begin{tabular}{c|cc} $L_1$ &1&2 \\ \hline 1&2&2 \\ 2&1&2 \end{tabular} \quad
\begin{tabular}{c|cccc} $L_2$ &1&2&3&4 \\ \hline 1&2&4&2&4 \\ 2&3&4&3&4 \\ 3&4&4&4&4 \\ 4&1&2&3&4 \end{tabular} \quad 
\begin{tabular}{c|cccc cccc} $L_3$ &1&2&3&4&5&6&7&8 \\ \hline 1&2&4&6&8&2&4&6&8 \\ 2&3&4&7&8&3&4&7&8 \\ 3&4&8&4&8&4&8&4&8 \\ 4&5&6&7&8&5&6&7&8 \\
                                                      5&6&8&6&8&6&8&6&8 \\           6&7&8&7&8&7&8&7&8 \\ 7&8&8&8&8&8&8&8&8 \\ 8&1&2&3&4&5&6&7&8
\end{tabular}  \]
Laver tables are also described in \cite{De00}. 
\par Many examples for LD-, bi-LD- and multi-LD-systems are given in Dehornoy's monograph \cite{De00}.

\subsection{$f$-conjugacy} 

\medskip

One may consider several generalizations of the conjugacy operation as candidates for natural LD-operations in groups.
Consider an ansatz like $x*y=f(x^{-1})g(y)h(x)$ for some group endomorphisms $f,g,h$.

\begin{prop} \label{PropLDConj}
Let $G$ be a group, and $f,g,h \in End(G)$. Then the binary operation
$x*y=f(x^{-1}) \cdot g(y)\cdot h(x)$ yields an LD-structure on $G$ if and only if 
\begin{equation} \label{fConjLDeqs} fh=f, \quad gh=hg=hf, \quad fg=gf=f^2, \quad h^2=h. \end{equation}
\end{prop}

\begin{proof}. A straightforward computation yields
\begin{eqnarray*}
\alpha * (\beta * \gamma )&=&f(\alpha ^{-1}) gf(\beta ^{-1}) g^2(\gamma ) gh(\beta ) h(\alpha ), \quad {\rm and} \\
(\alpha * \beta)*(\alpha *\gamma )&=& fh(\alpha ^{-1}) fg(\beta ^{-1}) f^2(\alpha ) gf(\alpha ^{-1}) g^2(\gamma ) gh(\alpha ) hf(\alpha ^{-1}) \cdot \\
&& hg(\beta ) h^2(\alpha ).
\end{eqnarray*}
A comparison of both terms yields the assertion.
\end{proof}

The simplest solution of the system of equations (\ref{fConjLDeqs}) is $f=g$ and $h={\rm id}$. This leads to the following definition. 
\begin{defi} ({\sc LD- or $f$-conjugacy}) \label{fConj}
Let $G$ be a group, and $f\in End(G)$. An ordered pair $(u, v)\in G \times G$ is called
$f$-LD-conjugated or LD-conjugated, or simply $f$-conjugated, denoted by $u\stackrel{}{\longrightarrow }_{*_f}v$, if there exists a $c\in G$ such that
$v=c*_f u=f(c^{-1}u)c$.
\end{defi}

\begin{rem}
For any non-trivial endomorphism $f$, the relation $\stackrel{}{\longrightarrow }_{*_f}$ defines not an equivalence relation on $G$.
Even the relation $\stackrel{}{\longrightarrow }_{*}$, defined by $u\stackrel{}{\longrightarrow }_*v$ if and only if there exists  an $f \in Aut(G)$ s.t. $u\stackrel{}{\longrightarrow }_{*_f}v$, is not an equivalence relation. Indeed, transitivity requires the automorphisms (relation must be symmetric!) to be an idempotent endomorphism ($f^2=f$) which implies $f={\rm id}$. 
\par
Compare the notion of $f$-LD-conjugacy with the well known notion {\it $f$-twisted conjugacy} defined by $u \sim _f v$ (for $f\in Aut(G)$) if and only if
there exists a $c\in G$ s.t. $v=f(c^{-1})uc=:c*^{tw}_f u$, which yields indeed an equivalence relation.
On the other hand, the operation $*^{tw}=*^{tw}_f$ is not LD - rather it satisfies the following "near" LD-law:
\[ \alpha *^{tw} (\beta *^{tw} \gamma )=(\alpha *^{tw} \beta)*^{tw}(\alpha ^f *^{tw}\gamma ), \]
where $\alpha ^f$ is short for $f(\alpha )$.  \par
Anyway, it follows directly from the definitions that $u\stackrel{}{\longrightarrow }_*v$ if and only if $f(u) \sim _f v$, i.e., any $f$-LD conjugacy problem
reduces to a twisted conjugacy problem and vice versa. Here we have to extend the notion of twisted conjugacy from $f\in Aut(G)$ to all $f\in End(G)$. 
\end{rem}

\begin{exa}
Recall that the $n$-strand braid group $B_n$ is generated by $\sigma _1$, ..., $\sigma _{n-1}$ where inside $\sigma _i$ the $(i+1)$-th
strand crosses over the $i$-th strand. There exists a natural epimorphism from $B_n$ onto the symmetric group $S_n$, defined by $\sigma _i \mapsto (i,i+1)$.
Let $G$ be the kernel of this epimorphism, namely the $n$-strand pure braid group $P_n$.
For some small integer $d\ge 1$, consider the epimorphism $\eta _d: P_n \longrightarrow P_{n-d}$ given by "pulling out" (or erasing) the last $d$ strands, i.e. the strands $n-d+1, \ldots , n$. Consider the shift map $\partial : B_{n-1} \longrightarrow B_n$, defined by $\sigma _i \mapsto \sigma _{i+1}$,
 and note that $\partial ^d (P_{n-d}) \le P_n$. Now, we define the endomorphism $f: P_n \longrightarrow P_n$ by the composition $f=\partial ^d \circ \eta _d$. 
\end{exa} 

\subsection{Shifted conjugacy}
Patrick Dehornoy introduced the following generalization of $f$-conjugacy, and he points out, that once the definition of shifted conjugacy is used, braids inevitably appear \cite{De00,De06}.

\begin{prop} \label{ExI.3.20.} {\rm (Exercise I.3.20. in \cite{De00})}
Consider a group $G$, a homomorphism $f: G\rightarrow G$, and a fixed element $a\in G$. Then the binary operation
\[ x*y=x *_{f,a}y=f(x)^{-1} \cdot a \cdot f(y)\cdot x  \]
yields an LD-structure on $G$ if and only if $[a,f^2(x)]=1$ for all $x\in G$, and $a$ satisfies the relation $a f(a) a=f(a) a f(a)$. \par
Hence the subgroup $H=\langle \{ f^n(a) \mid  n\in \mathbb{N} \} \rangle $ of $G$ is a homomorphic image of
the braid group
\[ B_{\infty }= \langle \{\sigma _i \}_{i\ge 1} \mid \sigma _i\sigma _j=\sigma _j\sigma _i \,\,{\rm for} \,\, |i-j|\ge 2, \,\, 
\sigma _i\sigma _j\sigma _i=\sigma _j\sigma _i\sigma _j \,\, {\rm for} \,\, |i-j|=1\rangle \] 
with infinitely many strands, i.e., up to an isomorphism, it is a quotient of $B_{\infty }$. 
\end{prop}

There exists a straightforward generalization of Proposition \ref{ExI.3.20.} for multi-LD-systems:

\begin{prop} \label{multiLD} Let $I$ be an index set. Consider a group $G$, a family of endomorphisms $(f_i)_{i\in I}$ of $G$, and a set of fixed elements $\{a_i\in G \mid i\in I\}$. 
Then $(G,(*_i)_{i\in I})$ with
\[ x*_iy= f_i(x^{-1})\cdot a_i \cdot f_i(y)\cdot x\]
is a multi-LD-system if and only if $f_i=f_j=:f$ for all $i\ne j$, $[a_i,f^2(x)]=1$ for all $x\in G$, $i\in I$, and $a_if(a_i)a_j=f(a_j)a_if(a_i)$ for all $i,j\in I$. 
\end{prop}

\begin{proof} A straightforward computation gives
\begin{eqnarray*} x*_i(y*_jz)&=&f_i(x^{-1})a_i [f_i(f_j(y^{-1})) f_i(a_j) f_i(f_j(z)) f_i(y)]x, \\
(x*_iy)*_j(x*_iz)&=& [f_j(x^{-1})  f_j(f_i(y^{-1}))  f_j(a_i^{-1})  f_j(f_i(x))]  a_j   [ f_j(f_i(x^{-1})) \cdot \\
 && f_j(a_i)  f_j(f_i(z))  f_j(x)][ f_i(x^{-1})  a_i f_i(y)  x].
\end{eqnarray*}
A comparison of both terms yields the assertion. 
\end{proof}

Note that this proof also contains proofs of Proposition \ref{ExI.3.20.} (setting $|I|=1$) and of the following Corollary \ref{ShConj} (setting $G=B_{\infty }$, $I=\{1,2\}$, $s=\partial $, $*_1=*$, $*_2=\bar{*}$, $a_1=\sigma _1$ and $a_2=\sigma _1^{-1}$). 

\medskip

Consider the injective {\it shift endomorphism} $\partial : B_{\infty } \longrightarrow B_{\infty }$ defined by $\sigma _i \mapsto \sigma _{i+1}$ for all $i\ge 1$.
\begin{coro} \label{ShConj} {\sc (Shifted conjugacy}, {\rm Example X.3.5. in \cite{De00})} 
$B_{\infty }$ equipped with the {\rm shifted conjugacy} operations $*$, $\bar{*}$ defined by
\[ x*y=\partial x^{-1}\cdot \sigma _1 \cdot \partial y \cdot x, \quad \quad  x\, \bar{*}\, y=\partial x^{-1}\cdot \sigma _1^{-1} \cdot \partial y  \cdot x  \]
is a bi-LD-system. In particular, $(B_{\infty },*)$ is an LD-system.
\end{coro}

\subsection{Generalized shifted conjugacy in braid groups}

In the following we consider generalizations of the shifted conjugacy operations $*$ in $B_{\infty }$. Therefore we
set $s=\partial ^p$ for some $p\in \mathbb{N}$, and we choose $a_i\in B_{2p}$ for all $i\in I$ such that 
\begin{equation} a_i\partial ^p(a_i)a_j=\partial ^p(a_j)a_i\partial ^p(a_i) \quad {\rm for} \,\, {\rm all} \,\,  i,j\in I.  \end{equation}  
Since $a_i\in B_{2p}$, we have $[a_i,\partial ^{2p}(x)]=1$ for all $x\in B_{\infty }$.
Thus the conditions of Proposition \ref{multiLD} are fulfilled, and $x*_iy=x\partial ^p(y)a_i\partial ^p(x^{-1})$ defines a multi-LD-structure on $B_{\infty}$. 
For $|I|=1$, $p=1$ and $a=\sigma _1$, which implies $H=B_{\infty }$, we get Dehornoy's original definition of shifted conjugacy $*$. 
\par
It remains to give some natural solutions $\{a_i\in B_{2p} \mid i\in I \}$ of the equation set (1). Note that in case $|I|=1$ (notation: $a_1=a$), of course, every
endomorphism $f$ of $B_{\infty }$ with $f(\sigma _1) \in B_{2p}$ provides such solution $a=f(\sigma _1)$. 

\begin{defi} (Definition I.4.6. in \cite{De00}) Let, for $n\ge 2$,  $\delta _n=\sigma _{n-1}\cdots \sigma _2\sigma _1$. For $p, q \ge 1$, we set
\[ \tau _{p,q}=\delta _{p+1}\partial (\delta _{p+1})\cdots \partial ^{q-1}(\delta _{p+1}). \]
\end{defi}

Since $a=\tau _{p,p}^{\pm 1}\in B_{2p}$ fulfills $a\partial ^p(a)a=\partial ^p(a)a\partial ^p(a)$, it provides a lot of (multi)-LD-structures on $B_{\infty }$.  

\begin{prop} \label{abc} (a) The binary operation $x*_ay=\partial ^p(x^{-1})a\partial ^p(y)x$ with $a=a'\tau _{p,p}a''$ for some $a',a''\in B_p$ yields an LD-structure on $B_{\infty }$ 
if and only if $[a',a'']=1$. 
\par
\noindent (b) Let $I$ be an index set.
The binary operations $x*_iy=\partial ^p(x^{-1})a_i\partial ^p(y)x$ with $a_i=a'_i\tau _{p,p}a''_i$ for some $a'_i,a''_i\in B_p$ ($i\in I$) yields a multi-LD-structure on $B_{\infty }$ 
if and only if $[a_i',a_j']=[a_i',a_j'']=1$ for all $i,j\in I$. (Note that $a_i''$ and $a_j''$ needn't commute for $i\ne j$.) 
\par
\noindent (c) The binary operations $x*_iy=\partial ^p(x^{-1})a_i\partial ^p(y)x$ ($i=1,2$) with $a_1=a_1'\tau _{p,p}a_1''$, $a_2=a_2'\tau _{p,p}^{-1}a_2''$ for some $a_1',a_1'',a_2',a_2''\in B_p$ 
yields a bi-LD-structure on $B_{\infty }$ if and only if $[a_1',a_1'']=[a_2',a_2'']=[a_1',a_2'']=[a_2',a_1'']=[a_1',a_2']=1$.
(Note that $a_1''$ and $a_2''$ needn't commute.)  
\end{prop}

We see that there exist infinitely many (multi)-LD-structures on $B_{\mathbb{N}}$. Further examples are provided by Proposition \ref{split}, which, of course, admits a lot of variations and generalizations.

\begin{prop} \label{split} Let be $p,p_1,p_2\in \mathbb{N}$ with $p_1+p_2=p$.  The binary operation $x*_ay=\partial ^p(x^{-1})a\partial ^p(y)x$ with 
\[ a=a_1'\partial ^{p_1}(a_2')\partial ^{p_1}(\tau _{p_2,p})\tau _{p,p_1}^{-1}a_1''\partial ^{p_1}(a_2'')\] for some $a_1',a_1''\in B_{p_1}$, $a_2',a_2''\in B_{p_2}$ yields an LD-structure 
on $B_{\infty }$ if and only if $[a_1',a_1'']=[a_2',a_2'']=1$. 
\end{prop}

The proofs of Proposition \ref{abc} and \ref{split} are straightforward computations. The reader is recommended to draw some pictures.

\subsection{Yet another group-based LD-system}
Though we are sure that it must have been well known to experts, we haven't been able to find the following natural LD-operation for groups in the literature.
For a group $G$, $(G,\circ )$ is an LD-system with 
\[ x\circ y=xy^{-1}x. \]
Note that, contrary to the conjugacy operation $*$, for this {\it "symmetric decomposition" or conjugacy operation} $\circ $, 
the corresponding relation $\stackrel{}{\longrightarrow }_{\circ }$, defined by $x\stackrel{}{\longrightarrow }_{\circ }y$ if and only if there exists a $c\in G$ such that $y=c\circ x$, is not an equivalence relation. In particular, $\stackrel{}{\longrightarrow }_{\circ }$ is reflexive and symmetric, 
but not transitive. 
\par 
One may consider several generalizations of this symmetric conjugacy operation $\circ $, as candidates for natural LD-operations in groups.
Consider an ansatz like $x\circ y=f(x)g(y^{-1})h(x)$ for some group endomorphisms $f,g,h$.

\begin{prop} \label{PropLDsymmConj}
Let $G$ be a group, and $f,g,h \in End(G)$. Then the binary operation
$x\circ y=f(x) \cdot g(y^{-1})\cdot h(x)$ yields an LD-structure on $G$ if and only if 
\begin{equation} \label{fsymmConjLDeqs} f^2=f, \quad fh=gh=fg, \quad hg=gf=hf, \quad h^2=h. \end{equation}
\end{prop}

\begin{proof}
A straightforward computation yields
\begin{eqnarray*}
\alpha \circ (\beta \circ \gamma )&=&f(\alpha ) gh(\beta ^{-1}) g^2(\gamma ) gf(\beta ^{-1}) h(\alpha ), \quad {\rm and} \\
(\alpha \circ \beta)\circ (\alpha \circ \gamma )&=& f^2(\alpha ) fg(\beta ^{-1}) fh(\alpha ) gh(\alpha ^{-1}) g^2(\gamma ) gf(\alpha ^{-1}) \cdot \\
&& hf(\alpha ) hg(\beta ^{-1}) h^2(\alpha ).
\end{eqnarray*}
A comparison of both terms yields the assertion. 
\end{proof}

Except for $f^2=f=g=h=h^2$, the simplest solutions of the system of equations (\ref{fsymmConjLDeqs}) are $f^2=f=g$ and $h={\rm id}$, or
$f={\rm id}$ and $g=h=h^2$.  
\begin{coro} ({\sc LD- or $f$-symmetric conjugacy}) \label{fsymC}
Let $G$ be a group, and $f\in End(G)$ an endomorphism that is also a projector ($f^2=f$).
Then $(G, \circ _f)$ and $(G, \circ _f^{\rm rev})$, defined by $x\circ _f y=f(xy^{-1})x$ and $x \circ _f^{\rm rev} y=xf(y^{-1}x)$, are LD-systems.
\end{coro}

\begin{prop}  \label{distrOversymmConj}
Let $G$ be a group, and $f, g \in End(G)$. \\
(i) Then the binary operations $\circ _f$ and $*_f$ (and $*^{\rm rev}_f$), defined by 
$x\circ _f y=f(x) \cdot g(y^{-1})\cdot h(x)$ and $x*_f y=f(x^{-1} \cdot y)\cdot h(x)$ ($x* ^{\rm rev}_f y=x \cdot f(y \cdot x^{-1})$),
are distributive over $\circ $. In particular $*$ ($*^{\rm rev}$) is distributive over $\circ $. In short, the following equations hold.
\[   x*_f(y\circ z)=(x*_fy)\circ (x*_fz), \quad x\circ _f(y\circ z)=(x\circ _fy)\circ (\circ _fz) \forall x,y,z\in G.\]
(ii) The operations $\circ _f$ and $*_f$ ($*^{\rm rev}_f$) are distributive over $\circ _g$ if and only if $f=gf=fg$.
\end{prop}

\section{Key establishment for all LD-systems}
\subsection{The protocol}

Recall that a \emph{magma} is a set $M$ equipped with a binary operation, say $\bullet $, which is possibly non-associative. 
For our purposes all interesting LD-systems are non-associative. 
Consider an element $y$ of a magma $(M,\bullet )$ which is an iterated product of other elements in $M$. Such an element can be described by a planar rooted binary
tree $T$ whose $k$ leaves are labelled by these other elements $y_1,\ldots ,y_k\in M$. We use the notation $y=T_{\bullet }(y_1,\ldots ,y_k)$. 
Here the subscript $\bullet $ tells us that the grafting of subtrees of $T$ corresponds to the operation $\bullet $. 
\par 
Consider, for example, the element $y=((b \bullet c) \bullet (a \bullet b)) \bullet b $. The corresponding
labelled planar rooted binary tree $T$ is displayed in the following figure.

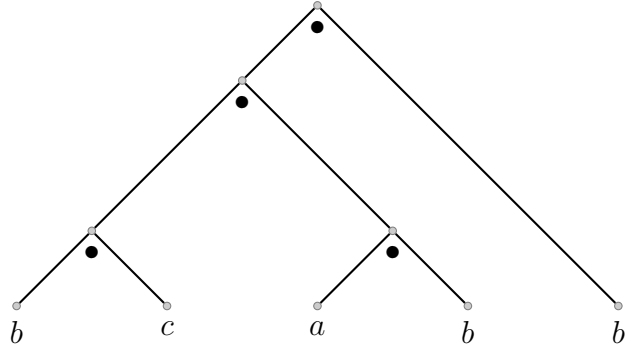
\begin{figure}[ht]
  \caption{ The element $y=((b \bullet c) \bullet (a \bullet b)) \bullet b =T_{\bullet }(b,c,a,b,b)$}
\begin{center}
\begin{tikzpicture}
   \node (r1) at (0,0) [circle,inner sep=1pt,draw=black!50,fill=black!20]{}; \node[below] at (r1.south) {$b$};  
   \node (r2) at (2,0) [circle,inner sep=1pt,draw=black!50,fill=black!20]{}; \node[below] at (r2.south) {$c$};  
   \node (r3) at (4,0) [circle,inner sep=1pt,draw=black!50,fill=black!20]{}; \node[below] at (r3.south) {$a$};  
   \node (r4) at (6,0) [circle,inner sep=1pt,draw=black!50,fill=black!20]{}; \node[below] at (r4.south) {$b$};  
   \node (r5) at (8,0) [circle,inner sep=1pt,draw=black!50,fill=black!20]{}; \node[below] at (r5.south) {$b$};  
   \node (i12) at (1,1) [circle,inner sep=1pt,draw=black!50,fill=black!20]{}
          edge[thick] (r1)   edge[thick] (r2);     \node[below] at (i12.south) {$\bullet $};  
   \node (i34) at (5,1) [circle,inner sep=1pt,draw=black!50,fill=black!20]{}
          edge[thick] (r3)   edge[thick] (r4);     \node[below] at (i34.south) {$\bullet $};  
   \node (i14) at (3,3) [circle,inner sep=1pt,draw=black!50,fill=black!20]{}
          edge[thick] (i12)   edge[thick] (i34);     \node[below] at (i14.south) {$\bullet $};  
   \node (i15) at (4,4) [circle,inner sep=1pt,draw=black!50,fill=black!20]{}
          edge[thick] (i14)   edge[thick] (r5);   \node[below] at (i15.south) {$\bullet $}; 
\end{tikzpicture}
\end{center}
\end{figure}

It is easy to prove by induction (over the depth of the involved trees) that any magma homomorphism $\beta :(M,\bullet )\rightarrow (N,\circ )$ satisfies
\[ \beta (T_{\bullet }(y_1,\ldots ,y_k))=T_{\circ }(\beta (y_1),\ldots ,\beta (y_k)) \]
for all $y_1,\ldots ,y_k\in M$. 

\begin{prop} \label{LDendo}
Let $(L,*)$ be an LD-system.
Then, for any element $x\in L$, the left multiplication map $\phi _x: y \mapsto x*y$ defines a magma endomorphism of $L$.
\end{prop}

\begin{proof} 
$\phi _x(y_1*y_2)=x*(y_1*y_2)\stackrel{LD}{=}(x*y_1)*(x*y_2)=\phi _x(y_1) * \phi _x(y_2)$. 
\end{proof}

We are going to describe a KEP that applies to any LD-system $(L,*)$.
There are two public submagmas $S_A=\langle s_1, \cdots , s_m \rangle _*$, $S_B=\langle t_1, \cdots , t_n \rangle _*$ of $(L,*)$, assigned to Alice and Bob.
Alice and Bob perform the following protocol steps.

\begin{description}
\item[{\bf Protocol 1}] {\sc Key establishment for any LD-system} $(L,*)$.
\item[{\rm 1}] Alice generates her secret key $(a_0,a) \in S_A \times L$, and Bob chooses his secret key $b\in S_B$.
\item[{\rm 2}] Alice computes the elements $a*t_1, \ldots ,a*t_n, p_0=a*a_0 \in L$, and sends them to Bob. 
Bob computes $b*s_1, \ldots ,b*s_m\in L$, and sends them to Alice. 
\item[{\rm 3}] Alice, knowing $a_0=T_*(r_1, \ldots , r_k)$ with $r_i\in \{s_1,\ldots ,s_m\}$, computes from the received message
\[ T_*(b*r_1, \ldots , b*r_k)=b*T_*(r_1, \ldots , r_k)=b*a_0. \]
And Bob, knowing $b=T'_*(u_1, \ldots , u_{k'})$ with $u_j\in \{t_1,\ldots ,t_n\}$, computes from his received message
\[ T'_*(a*u_1, \ldots , a*u_{k'})=a*T'_*(u_1, \ldots , u_{k'})=a*b. \]
\item[{\rm 4}] Alice computes $K_A=a*(b*a_0)$.
Bob gets the shared key by 
\[ K_B:=(a*b)*p_0=(a*b)*(a*a_0)\stackrel{(LD)}{=}K_A. \]
\end{description}

This protocol is an asymmetric modification of the Anshel-Anshel-Goldfeld protocols for magmas introduced in \cite{Ka07, Ka12}. 

\begin{figure}[ht]
  \caption{\sc Protocol 1: Key establishment for any LD-system}
\begin{center} 
 \begin{tikzpicture}
 \node[red]  (Alice) at ( 0,0) [circle,draw=black!50,fill=black!20]{Alice};
 \node[blue]  (Bob) at ( 7,0) [circle,draw=black!50,fill=black!20]{Bob}
   edge [<-, bend right=10] node[auto,swap] (pA) {$\{ {\color{red}a }* {\color{green}t_i} \}_{1\le i \le n}, \, {\color{red}a*a_0} $} (Alice)   
   edge [->, bend left=10] node[auto] (pB) { $\{ {\color{red}b } *{\color{green}s_j} \}_{1\le j \le m} $} (Alice) ;
 \node (skA) [below] at (Alice.south) [rectangle,draw=red!50,fill=red!20]{${\color{red}a_0}\in S_A, {\color{red}a}$};
 \node (skB) [below] at (Bob.south) [rectangle,draw=red!50,fill=red!20]{${\color{red}b}\in S_B$};
 \begin{pgfonlayer}{background}
  \node [fill=green!20, rounded corners, fit=(Alice) (Bob) (skA) (skB) (pA) (pB)] {};
 \end{pgfonlayer} 
 \end{tikzpicture}
\end{center}
\end{figure}
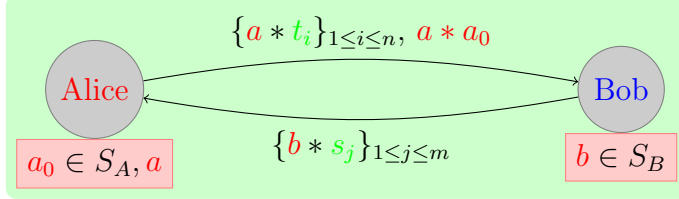

\subsection{Base problems} 
In order to break Protocol 1 an attacker has to find the shared key $K=K_A=K_B$.
A successful attack on Bob's secret key $b$ requires the solution of
\begin{list}{}{\setlength{\itemsep}{0cm} \setlength{\parsep}{0cm} }
\item[{\bf $m$-simLDP} ($m$-simultaneous LD-Problem):]
\item[{\sc Input:}]  Element pairs $(s_1,s'_1),\ldots ,(s_m,s'_m)\in L^2$ with $s'_i=b*s_i$ $\forall 1\le i\le m$ for some (unknown) $b \in L$.  
\item[{\sc Objective:}] Find $b'\in L$ with $b'*s_i=s'_i$ for all $i=1,\ldots ,m$.
\end{list}

Note that in our context, $b$ comes from a restricted domain, namely $S_A \subseteq L$. This might affect distributions when one considers possible attacks.
Nevertheless, we use the notion of (simultaneous) LD-Problem for inputs generated by potentially arbitrary $b\in L$. Similar remarks affect
base problems further in the text. 
\par 
Even if an attacker finds Bob's original key $b$ or a pseudo-key $b'$ (solution to the $m$-simLDP above), then she still faces the following problem.
\begin{list}{}{\setlength{\itemsep}{0cm} \setlength{\parsep}{0cm} }
\item[{\bf $*$-MSP} ($*$-submagma Membership Search Problem):]
\item[{\sc Input:}] $t_1,\ldots ,t_n\in (L,*)$, $b\in \langle t_1, \ldots , t_n \rangle _*$.
\item[{\sc Objective:}] Find an expression of $b$ as a tree-word in the submagma $\langle t_1,\ldots ,t_n \rangle _*$ (notation $b=T_*(u_1,\ldots ,u_k)$ for $u_i \in \{ t_j \} _{j\le n}$).   
\end{list}

\begin{prop}  \label{BaseProbs1KEP1}
Let $(L,*)$ be an LD-system. We define the generalized $m$-simLDP for $S_B\subseteq L$ as an $m$-simultaneous LD-Problem 
with the objective to find a $b'$ in $S_B=\langle t_1, \ldots , t_n \rangle _*$ such that $b'*s_i=s'_i$ for all $i\le m$. 
\par An oracle that solves the generalized $m$-simLDP and $*$-MSP for $S_B$ is sufficient to break key establishment Protocol 1.
\end{prop}

\begin{proof} 
As outlined above, we perform an attack on Bob's private key. The generalized $m$-simLDP oracle provides a pseudo-key $b'\in S_B$ with  $b'*s_i=s'_i=b*s_i$ for all $i=1,\ldots ,m$. Observe that this implies for any element $e_A\in S_A$ that $b'*e_A=b*e_A$.
In particular, we have $b'*a_0=b*a_0$.
We feed this pseudo-key $b'$ into a $*$-MSP oracle for $S_B$ which returns a treeword  $T'_*(u_1,\ldots ,u_l)=b'$ (for some $l\in \mathbb{N}$ and $u_i \in \{ t_j \} _{j\le n}$). Now compute 
\begin{eqnarray*}
T'_*(a*u_1,\ldots ,a*u_l)*p_0 & \stackrel{LD}{=} & (a*T'_*(u_1,\ldots ,u_l))*p_0 = (a*b')*(a*a_0) \\
& \stackrel{LD}{=} &a*(b'*a_0)=a*(b*a_0)=K.  
\end{eqnarray*}
\end{proof}

Note that here the situation is asymmetric - an attack on Alice's secret key requires the solution of the following problem.
\begin{list}{}{\setlength{\itemsep}{0cm} \setlength{\parsep}{0cm} }
\item[{\bf $n$-modsimLDP} (Modified $n$-Simultaneous LD-Problem):]
\item[{\sc Input:}]  An element $p_0\in L$ and pairs $(t_1,t'_1),\ldots ,(t_n,t'_n)\in L^2$ with $t'_i=a*t_i$ $\forall 1\le i\le n$ for some (unknown) $a \in L$.  
\item[{\sc Objective:}] Find elements $a'_0, a'\in L$ such that $p_0=a'*a'_0$ and $a'*t_i=t'_i$ for all $i=1,\ldots ,n$.
\end{list}

Also here, even if an attacker finds Alice's original key $(a_0,a)$ or a pseudo-key $(a'_0, a') \in S_A \times L$, 
then she still faces a $*$-submagma Membership Search Problem.

\begin{prop} \label{BaseProbs2KEP1}
Let $(L,*)$ be an LD-system. We define the generalized $n$-modsimLDP for $S_A\subseteq L$ as a modified $n$-simultaneous LD-Problem 
with the objective to find $a'\in L$ and $a'_0$ in $S_A=\langle s_1, \ldots , s_m \rangle _*$ such that $a'*t_i=t'_i$ for all $i\le n$. 
\par An oracle that solves the generalized $n$-modsimLDP and $*$-MSP for $S_A$ is sufficient to break key establishment Protocol 1.
\end{prop}

\begin{proof} 
As outlined above, we perform an attack on Alice's private key. The generalized $n$-simLDP oracle provides a pseudo-key $(a'_0, a')'\in S_A \times L$ such that $a'*a'_0=p_0$ and $a'*t_i=a'_i=a*t_i$ for all $i=1,\ldots ,n$. Observe that this implies for any element $e_B\in S_B$ that $a'*e_B=a*e_B$.
In particular, we have $a'*b=a*b$.
We feed the first component $a'_0 \in S_A$ of this pseudo-key into a $*$-MSP oracle for $S_A$ which returns a treeword  $T'_*(r_1,\ldots ,r_l)=a'_0$ (for some $l\in \mathbb{N}$ and $r_i \in \{ s_j \} _{j\le m}$). Now, we compute 
\begin{eqnarray*}
a'*T'_*(b*r_1,\ldots ,b*r_l) & \stackrel{LD}{=} & a'*(b*T'_*(r_1,\ldots ,r_l))= a'*(b*a'_0) \\
& \stackrel{LD}{=} &(a'*b)*(a'*a'_0)=(a*b)*p_0=K.   
\end{eqnarray*}
\end{proof}

Both appproaches described above require the solution of a $*$-submagma Membership Search Problem. Note that we assumed that the generalized $m$-simLDP 
(resp. $n$-modsimLDP) oracle already provides a pseudo-key in the submagma $S_B$ (resp. $S_A$) which we feed to the $*$-MSP oracle.
But to check whether an element lies in some submagma, i.e. the $*$-submagma Membership Decision Problem, is already undecidable in general. 
\par Fortunately, for the attacker, there are approaches which do not resort to solving the $*$-MSP. 

\medskip

Recall that we defined the generalized $m$-simLDP for $S_B\subseteq L$ as an $m$-simulta\-neous LD-Problem 
with the objective to find a $b'$ in $S_B=\langle t_1, \ldots , t_n \rangle _*$ such that $b'*s_i=s'_i$ for all $i\le m$.
\begin{prop}  \label{BaseProbs3KEP1} 
A generalized simLDP oracle is sufficient to break key establishment Protocol 1.
More precisely, an oracle that solves the generalized $m$-simLDP for $S_B$ {\it and} the $n$-simLDP is sufficient to break Protocol 1.
\end{prop}

\begin{proof} 
Here we perform attacks on Alice's and Bob's private keys - though we need only a pseudo-key for the second component $a'$ of Alice's key. 
The $n$-simLDP oracle provides $a'\in L$ s.t. $a'*t_j=t'_j=a*t_j$ for all $j\le n$. 
And the generalized $m$-simLDP oracle returns the pseudo-key $b'\in S_B$
s.t. $b'*s_i=s'_i=b*s_i$ for all $i\le m$. Since $b'\in S_B$, we conclude that $a'*b'=a*b'$.
Also, $a_0\in S_A$ implies, of course, $b'*a_0=b*a_0$.
Now, we may compute 
\[ (a'*b')*p_0 =(a*b')*(a*a_0)\stackrel{LD}{=}a*(b'*a_0)=a*(b*a_0)=K. \]
\end{proof}

Recall that we defined the generalized $n$-modsimLDP for $S_A\subseteq L$ as an $n$-simultaneous LD-Problem 
with the objective to find a $a'_0$ in $S_A=\langle s_1, \ldots , s_m \rangle _*$ such that $a'*t_i=t'_i$ for all $i\le n$.

\begin{prop}  \label{BaseProbs4KEP1} 
An oracle that solves the generalized $n$-modsimLDP for $S_A$ {\it and} the $m$-simLDP is sufficient to break Protocol 1.
\end{prop} 

\begin{proof}
Also here we perform attacks on Alice's and Bob's private keys.  
The $m$-simLDP oracle provides $b'\in L$ s.t. $b'*s_j=s'_j=b*s_j$ for all $j\le m$. 
And the generalized $n$-modsimLDP oracle returns the pseudo-key $(a'_0, a')\in S_A \times L$
s.t. $a'*t_i=t'_i=a*t_i$ for all $i\le n$ {\it and} $a'*a'_0=p_0$. Since $a'_0\in S_A$, we conclude that $b'*a'_0=b*a'_0$.
Also, $b\in S_B$ implies, of course, $a'*b=a*b$.
Now, we compute 
\[ a'*(b'*a'_0) =a'*(b*a'_0)\stackrel{LD}{=}(a'*b)*(a'*a'_0)=(a*b)*p_0=K.   \]
\end{proof}

\section{Key establishment for left distributive systems}
\subsection{The protocol}  \label{KEPpartial}
Here we describe a generalization of Protocol 1 that works for all multi-LD-systems.
Actually, it suffices if $L$ is only a partial multi-LD-system, i.e. some distributive laws hold.
More precisely, consider a set $L$ equipped with a pool of binary operations $O_A \cup O_B$ ($O_A$ and $O_B$ non-empty) s.t.
the operations in $O_A$ are distributive over those in $O_B$ and vice versa, i.e. the following holds
for all $x,y,z\in L$, $*_{\alpha } \in O_A$ and $*_{\beta }\in O_B$.
\begin{eqnarray}
  x*_{\alpha }(y*_{\beta }z)&=&(x*_{\alpha }y)*_{\beta }(x*_{\alpha }z), \,\, {\rm and}    \label{abLD} \\
  x*_{\beta }(y*_{\alpha }z)&=&(x*_{\beta }y)*_{\alpha }(x*_{\beta }z).                    \label{baLD}
\end{eqnarray}
Note that, if $O_A \cap O_B \ne \emptyset$, then $(L, O_A \cup O_B)$ is a multi-LD-system. 
\par Let $s_1, \ldots , s_m, t_1, \ldots , t_n\in L$ be some public elements. We denote
$S_A=\langle s_1, \cdots , s_m \rangle _{O_A}$ and $S_B=\langle t_1, \cdots , t_n \rangle _{O_B}$, two submagmas of $(L,O_A\cup O_B)$.
For example, an element $y$ of $S_A$ can be described by a planar rooted binary
tree $T$ whose $k$ leaves are labelled by these other elements $r_1,\ldots ,r_k$ with $r_i \in \{s_i\}_{i\le m}$.
Here the tree contains further information, namely to each internal vertex we assign a binary operation $*_i \in O_A$.
We use the notation $y=T_{O_A}(r_1,\ldots ,r_k)$. 
The subscript $O_A $ tells us that the grafting of subtrees of $T$ corresponds to the operation $*_i\in O_A$.
Consider, for example, the element $y=((s_3*_{\alpha _2}s_3)*_{\alpha _4}s_1)*_{\alpha _1}(s_2*_{\alpha _2}s_1)$. The corresponding
labelled planar rooted binary tree $T$ is displayed in the following figure.

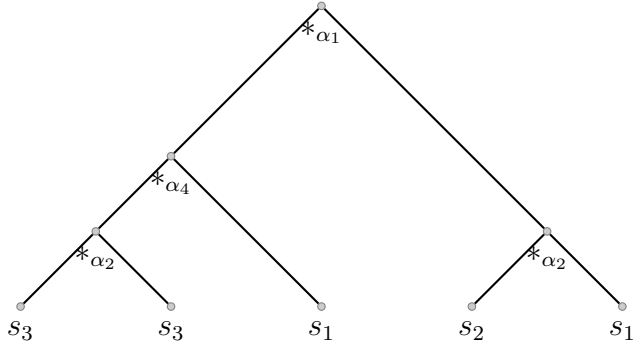
\begin{figure}[ht]
  \caption{ The element $y=((s_3*_{\alpha _2}s_3)*_{\alpha _4}s_1)*_{\alpha _1}(s_2*_{\alpha _2}s_1) \in S_A$}
\begin{center}
\begin{tikzpicture}
   \node (r1) at (0,0) [circle,inner sep=1pt,draw=black!50,fill=black!20]{}; \node[below] at (r1.south) {$s_3$};  
   \node (r2) at (2,0) [circle,inner sep=1pt,draw=black!50,fill=black!20]{}; \node[below] at (r2.south) {$s_3$};  
   \node (r3) at (4,0) [circle,inner sep=1pt,draw=black!50,fill=black!20]{}; \node[below] at (r3.south) {$s_1$};  
   \node (r4) at (6,0) [circle,inner sep=1pt,draw=black!50,fill=black!20]{}; \node[below] at (r4.south) {$s_2$};  
   \node (r5) at (8,0) [circle,inner sep=1pt,draw=black!50,fill=black!20]{}; \node[below] at (r5.south) {$s_1$};  
   \node (i12) at (1,1) [circle,inner sep=1pt,draw=black!50,fill=black!20]{}
          edge[thick] (r1)   edge[thick] (r2);     \node[below] at (i12.south) {$*_{\alpha _2}$};  
   \node (i13) at (2,2) [circle,inner sep=1pt,draw=black!50,fill=black!20]{}
          edge[thick] (i12)   edge[thick] (r3);     \node[below] at (i13.south) {$*_{\alpha _4}$};  
   \node (i45) at (7,1) [circle,inner sep=1pt,draw=black!50,fill=black!20]{}
          edge[thick] (r4)   edge[thick] (r5);     \node[below] at (i45.south) {$*_{\alpha _2}$};  
   \node (i15) at (4,4) [circle,inner sep=1pt,draw=black!50,fill=black!20]{}
          edge[thick] (i13)   edge[thick] (i45);   \node[below] at (i15.south) {$*_{\alpha _1}$}; 
\end{tikzpicture}
\end{center}
\end{figure}

Let $*_{\alpha }\in O_A$ and $*_{\beta }\in O_B$. By induction over the tree depth, it is easy to show that, for all elements $e, e_1, \ldots , e_l \in (L, O_A \cup O_B)$ and all planar rooted binary trees $T$ with $l$ leaves, the following equations hold.
\begin{eqnarray}
e*_{\alpha }T_{O_B}(e_1, \ldots , e_l)&=&T_{O_B}(e*_{\alpha }e_1, \ldots , e*_{\alpha }e_l),   \\
e*_{\beta }T_{O_A}(e_1, \ldots , e_l)&=&T_{O_A}(e*_{\beta }e_1, \ldots , e*_{\beta }e_l).
\end{eqnarray}

Now, we are going to describe a KEP that applies to any system $(L,O_A\cup O_B)$ as described above.
We have two subsets of public elements $\{ s_1, \cdots , s_m \}$ and $\{t_1, \cdots , t_n \}$ of $L$.
Also, recall that $S_A=\langle s_1, \cdots , s_m \rangle _{O_A}$ and $S_B=\langle t_1, \cdots , t_n \rangle _{O_B}$.
Alice and Bob perform the following protocol steps.

\begin{description}
\item[{\bf Protocol 2}] {\sc Key establishment for the partial multi-LD-system} \par
$(L,O_A\cup O_B)$.
\item[{\rm 1}] Alice generates her secret key $(a_0,a, *_{\alpha }) \in S_A \times L \times O_A$, and Bob chooses his secret key 
   $(b, *_{\beta })\in S_B \times O_B$.
\item[{\rm 2}] Alice computes the elements $a*_{\alpha }t_1, \ldots ,a*_{\alpha }t_n, p_0=a*_{\alpha }a_0 \in L$, and sends them to Bob. 
Bob computes $b*_{\beta }s_1, \ldots ,b*_{\beta }s_m\in L$, and sends them to Alice. 
\item[{\rm 3}] Alice, knowing $a_0=T_{O_A}(r_1, \ldots , r_k)$ with $r_i\in \{s_1,\ldots ,s_m\}$, computes from Bob's public key
\[ T_{O_A}(b*_{\beta }r_1, \ldots , b*_{\beta }r_k)=b*_{\beta }T_{O_A}(r_1, \ldots , r_k)=b*_{\beta }a_0. \]
And Bob, knowing $b=T'_{O_B}(u_1, \ldots , u_{k'})$ with $u_j\in \{t_1,\ldots ,t_n\}$, computes from Alice's public key
\[ T'_{O_B}(a*_{\alpha }u_1, \ldots , a*_{\alpha }u_{k'})=a*_{\alpha }T'_{O_B}(u_1, \ldots , u_{k'})=a*_{\alpha }b. \]
\item[{\rm 4}] Alice computes $K_A=a*_{\alpha }(b*_{\beta }a_0)$.
Bob gets the shared key by 
\[ K_B:=(a*_{\alpha }b)*p_0=(a*_{\alpha }b)*_{\beta }(a*_{\alpha }a_0)\stackrel{LD}{=}K_A.  \]
\end{description}

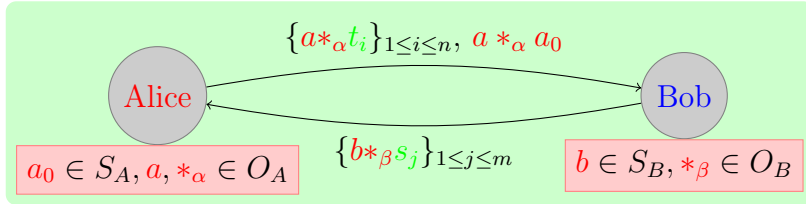
\begin{figure}[ht]
  \caption{{\sc KEP for the partial multi-LD-system} $(L,O_A\cup O_B)$.}
\begin{center} 
 \begin{tikzpicture}
 \node[red]  (Alice) at ( 0,0) [circle,draw=black!50,fill=black!20]{Alice};
 \node[blue]  (Bob) at ( 7,0) [circle,draw=black!50,fill=black!20]{Bob}
   edge [<-, bend right=10] node[auto,swap] (pA) {$\{ {\color{red}a *_{\alpha }} {\color{green}t_i} \}_{1\le i \le n}, \, {\color{red}a*_{\alpha }a_0} $} (Alice)   
   edge [->, bend left=10] node[auto] (pB) { $\{ {\color{red}b *_{\beta }} {\color{green}s_j} \}_{1\le j \le m} $} (Alice) ;
 \node (skA) [below] at (Alice.south) [rectangle,draw=red!50,fill=red!20]{${\color{red}a_0}\in S_A, {\color{red}a}, {\color{red}*_{\alpha }} \in O_A$};
 \node (skB) [below] at (Bob.south) [rectangle,draw=red!50,fill=red!20]{${\color{red}b}\in S_B, {\color{red}*_{\beta }} \in O_B$};
 \begin{pgfonlayer}{background}
  \node [fill=green!20, rounded corners, fit=(Alice) (Bob) (skA) (skB) (pA) (pB)] {};
 \end{pgfonlayer} 
 \end{tikzpicture}
\end{center}
\end{figure}

Here the operations $*_{\alpha } \in O_A$ and $*_{\beta }\in O_B$ are part of Alice's and Bob's private keys.
As in Protocol 1, explicit expressions of $a_0\in S_A$ and $b \in S_B$ as treewords $T,T'$ are also parts of the private keys - though we did not 
mention it explicitly in step 1 of the protocols. But here $T_{O_A}$ and $T'_{O_B}$ also contain all the information about the grafting operations
(in $O_A$ or $O_B$, respectively) at the internal vertices of $T$, $T'$.

\subsection{Base problems}

In order to break Protocol 2 an attacker has to find the shared key $K=K_A=K_B$.
A successful attack on Bob's secret key $(b, *_{\beta })$ requires (first) the solution of the following problem.
\begin{list}{}{\setlength{\itemsep}{0cm} \setlength{\parsep}{0cm} }
\item[{\sc Input:}]  Element pairs $(s_1,s'_1),\ldots ,(s_m,s'_m)\in L^2$ with $s'_i=b*_{\beta }s_i$ $\forall 1\le i\le m$ for some (unknown) $b \in L$,
   $*_{\beta }\in O_B$.  
\item[{\sc Objective:}] Find $b'\in L$ and   $*_{\beta '}\in O_B$ such that $b'*_{\beta '}s_i=s'_i$ for all $i=1,\ldots ,m$.
\end{list}

In order to clarify concepts we introduce the following notation, which also makes it easier to name our base problems at hand.
For $e\in L$, let $\phi _{e, \alpha }(x):=e*_{\alpha }x$. Then, for $*_{\alpha }\in O_A$, $\phi _{e, \alpha }$ is by (\ref{abLD}) a magma homomorphism on $S_B$.
Analogeously, for $*_{\beta }\in O_B$, $\phi _{e, \beta } \in End(S_A)$ by (\ref{baLD}).
Now, we may reformulate the base problem for obtaining a pseudo-key on Bob's secret.
\begin{list}{}{\setlength{\itemsep}{0cm} \setlength{\parsep}{0cm} }
\item[{\bf LDEndP} (LD-endomorphism Search Problem for $S_A$):]
\item[{\sc Input:}]  Element pairs $(s_1,s'_1),\ldots ,(s_m,s'_m)\in L^2$ with $s'_i=\phi _{b, \beta }(s_i)$ $\forall 1\le i\le m$ for some (unknown) 
magma endomorphism $\phi _{b, \beta }\in End(S_A)$ (with $*_{\beta }\in O_B$).  
\item[{\sc Objective:}] Find magma endomorphism $\phi _{b', \beta '}\in End(S_A)$ ($*_{\beta '}\in O_B$) such that $\phi _{b', \beta '}(s_i)=s'_i$ for all $i=1,\ldots ,m$.
\end{list}

Recall that we work in the leftdistributive system $(L, O_A \cup O_B)$.
We define the {\it generalized LDEndP} for $(S_A, S_B)$ as an LD-endomorphism Search Problem for $S_A$ 
with the objective to find a magma endomorphism $\phi _{b', \beta '}\in End(S_A)$ with $*_{\beta '}\in O_B$ {\it and} $b'$ in $S_B=\langle t_1, \ldots , t_n \rangle _{O_B}$.  
\par
Even if an attacker finds a pseudo-key endomorphism $\phi_{b',\beta '} \in End(S_A)$, then she still faces the following problem.
\begin{list}{}{\setlength{\itemsep}{0cm} \setlength{\parsep}{0cm} }
\item[{\bf $O_B$-MSP} ($O_B$-submagma Membership Search Problem for $S_B$):]
\item[{\sc Input:}] $t_1,\ldots ,t_n\in L$, $b\in S_B=\langle t_1, \ldots , t_n \rangle _{O_B}$.
\item[{\sc Objective:}] Find an expression of $b$ as a tree-word (with internal vertices labelled by operations in $O_B$) in the submagma $S_B$ (notation \\ $b=T_{O_B}(u_1,\ldots ,u_k)$ for $u_i \in \{ t_j \} _{j\le n}$).  
\end{list}

\begin{prop}  \label{BaseProbs1KEP2}
An oracle that solves the generalized LDEndP for $(S_A, S_B)$ and $O_B$-MSP for $S_B$ is sufficient to break key establishment Protocol 2.
\end{prop}

\begin{proof}
As outlined above, we perform an attack on Bob's private key. The generalized LDEndP for $(S_A, S_B)$ oracle provides a pseudo-key endomorphism
$\phi _{b', \beta '}\in End(S_A)$ with  $b'\in S_B$, $*_{\beta '}\in O_B$ such that  $\phi _{b', \beta '}(s_i)=s'_i=\phi _{b, \beta }(s_i)$
for all $i=1,\ldots ,m$. Observe that this implies for any element $e_A\in S_A$ that $\phi _{b', \beta '}(e_A)=\phi _{b, \beta }(e_A)$.
In particular, we have $\phi _{b', \beta '}(a_0)=\phi _{b, \beta }(a_0)$.
Since $b'\in S_B$, we may feed $b'$ into a $O_B$-MSP oracle for $S_B$ which returns a tree-word  $T'_{O_B}(u_1,\ldots ,u_l)=b'$ (for some $l\in \mathbb{N}$ and $u_i \in \{ t_j \} _{j\le n}$). Now, we compute 
\begin{eqnarray*}
&& T'_{O_B}(a*_{\alpha }u_1,\ldots ,a*_{\alpha }u_l)*_{\beta '}p_0 \stackrel{LD}{=}  (a*_{\alpha }T'_{O_B}(u_1,\ldots ,u_l))*_{\beta '}p_0 \\
&=& (a*_{\alpha }b')*_{\beta '}(a*_{\alpha }a_0) 
 \stackrel{LD}{=} a*_{\alpha }(b'*_{\beta '}a_0)=a*_{\alpha }(b*_{\beta }a_0)=K.   
\end{eqnarray*}
\end{proof}

On the other hand, an attack on Alice's secret key requires (first) the solution of the following problem.
\begin{list}{}{\setlength{\itemsep}{0cm} \setlength{\parsep}{0cm} }
\item[{\bf modLDEndP} (Modified LD-endomorphism Search Problem for $S_B$):]
\item[{\sc Input:}]  Element pairs $(t_1,t'_1),\ldots ,(t_n,t'_n)\in L^2$ with $t'_i=\phi _{a, \alpha }(t_i)$ $\forall 1\le i\le n$ for some (unknown) 
magma endomorphism $\phi _{a, \alpha }\in End(S_B)$ (with $*_{\alpha }\in O_A$). Furthermore, an element $p_0 \in \phi _{a, \alpha }(S_A)$, i.e.
$p_0=\phi _{a, \alpha }(a_0)$ for some $a_0\in S_A$.
\item[{\sc Objective:}] Find $(a'_0, \phi _{a', \alpha '})\in L \times End(S_B)$ ($*_{\alpha }\in O_A$) such that $\phi _{a', \alpha '}(t_i)=t'_i$ for all $i=1,\ldots ,n$ {\it and} $\phi _{a', \alpha '}(a'_0)=p_0$.
\end{list}

We define the {\it generalized modLDEndP} for $(S_B, S_A)$ as a modified LD-endo\-morphism Search Problem for $S_B$ 
with the objective to find $(a'_0, \phi _{a', \alpha '})\in S_A \times End(S_B)$ ($*_{\alpha }\in O_A$) such that $\phi _{a', \alpha '}(t_i)=t'_i$ for all $i=1,\ldots ,n$ {\it and} $\phi _{a', \alpha '}(a'_0)=p_0$.  
\par Even if an attacker finds a pseudo-key $(a'_0, \phi _{a', \alpha '}) \in S_A \times End(S_B)$ for Alice's secret, 
then she still faces a $O_A$-submagma Membership Search Problem for $S_A$.

\begin{prop}  \label{BaseProbs2KEP2}
An oracle that solves the generalized modLDEndP for 
$(S_B,$ \\ $S_A)$ and $O_A$-MSP for $S_A$ is sufficient to break key establishment Protocol 1.
\end{prop}

\begin{proof}
As outlined above, we perform an attack on Alice's private key. The generalized modLDEndP oracle provides a pseudo-key 
$(a'_0, \phi _{a', \alpha '}) \in S_A \times End(S_B)$
such that $\phi _{a', \alpha '}(t_i)=t'_i=\phi _{a, \alpha }(t_i)$ for all $i=1,\ldots ,n$ {\it and} $\phi _{a', \alpha '}(a'_0)=p_0$. 
Observe that this implies for any element $e_B\in S_B$ that 
 $\phi _{a', \alpha '}(e_B)=\phi _{a, \alpha }(e_B)$.
In particular, we have $\phi _{a', \alpha '}(b)=\phi _{a, \alpha }(b)$.
Since $a'_0 \in S_A$, we may feed $a'_0$ into a $O_A$-MSP oracle for $S_A$ which returns a tree-word  
$T'_{O_A}(r_1,\ldots ,r_l)=a'_0$ (for some $l\in \mathbb{N}$ and $r_i \in \{ s_j \} _{j\le m}$). Now, we may compute 
\begin{eqnarray*}
&& a'*_{\alpha '}T'_{O_A}(b*_{\beta }r_1,\ldots ,b*_{\beta }r_l)  \stackrel{LD}{=}  a'*_{\alpha '}(b*_{\beta }T'_{O_A}(r_1,\ldots ,r_l))  \\
&=& a'*_{\alpha '}(b*_{\beta }a'_0) \stackrel{LD}{=} (a'*_{\alpha '}b)*_{\beta }(a'*_{\alpha '}a'_0)=(a*_{\alpha }b)*_{\beta }p_0=K.   
\end{eqnarray*}
\end{proof}

Now, we describe approaches to break Protocol 2 which do not resort to solving a submagma-MSP. 

\begin{prop} \label{BaseProbs3KEP2} 
A generalized LDEndP oracle is sufficient to break key establishment Protocol 2.
More precisely, an oracle that solves the generalized LDEndP for $(S_A, S_B)$ {\it and} the LDEndP for $S_B$ is sufficient to break KEP1.
\end{prop}

\begin{proof}
Here we perform attacks on Alice's and Bob's private keys - though we do not require a pseudo-key for the first component $a_0$ of Alice's key. 
The LDEndP oracle for $S_B$ provides $\phi _{a', \alpha '}$ s.t. $\phi _{a', \alpha '}(t_j)=t'_j=\phi _{a, \alpha }(t_j)$ for all $j\le n$. 
And the generalized LDEndP oracle for $(S_A,S_B)$ returns the pseudo-key endomorphism $\phi _{b', \beta '}$ with $b'\in S_B$
s.t. $\phi _{b', \beta '}(s_i)=s'_i=\phi _{b, \beta }(s_i)$ for all $i\le m$. Since $b'\in S_B$, 
we conclude that $\phi _{a', \alpha '}(b')=\phi _{a, \alpha }(b')$.
Also, $a_0\in S_A$ implies, of course, $\phi _{b', \beta '}(a_0)=\phi _{b, \beta }(a_0)$.
Now, we compute 
\[ (a'*_{\alpha '}b')*_{\beta '}p_0 =(a*_{\alpha }b')*_{\beta '}(a*_{\alpha }a_0)\stackrel{LD}{=}a*_{\alpha }(b'*_{\beta '}a_0)=a*_{\alpha }(b*_{\beta }a_0)=K. \]
\end{proof}

Alternatively, one may choose the following approach.
\begin{prop}  \label{BaseProbs4KEP2} 
An oracle that solves the generalized modLDEndP for
$(S_B,$ \\ $S_A)$ {\it and} the LDEndP for $S_A$ is sufficient to break KEP1.
\end{prop} 

\begin{proof}
Also here we perform attacks on Alice's and Bob's private keys.  
The LDEndP oracle for $S_A$ provides $\phi _{b', \beta '}\in End(S_A)$ s.t. $\phi _{b', \beta '}(s_j)=s'_j=\phi _{b', \beta '}(s_j)$ for all $j\le m$. 
And the generalized modLDEndP oracle for $(S_B,S_A)$ returns the pseudo-key $(a'_0, \phi _{a', \alpha '}) \in S_A \times End(S_B)$
s.t. $\phi _{a', \alpha '}(t_i)=t'_i=\phi _{a', \alpha '}(t_i)$ for all $i\le n$ {\it and} $\phi _{a', \alpha '}(a'_0)=p_0$. Since $a'_0\in S_A$, we conclude that $\phi _{b', \beta '}(a'_0)=\phi _{b, \beta }(a'_0)$.
Also, $b\in S_B$ implies, of course, $\phi _{a', \alpha '}(b)=\phi _{a, \alpha }(b)$.
Now, we compute 
\[ a'*_{\alpha '}(b'*_{\beta '}a'_0) =a'*_{\alpha '}(b*_{\beta }a'_0)\stackrel{LD}{=}(a'*_{\alpha '}b)*_{\beta }(a'*_{\alpha '}a'_0)=
(a*_{\alpha }b)*_{\beta }p_0=K.   \]
\end{proof}

\begin{rem}
Note that in the non-associative setting the case $m=n=1$ is of particular interest, i.e. we may \emph{abandon simultaneity} in our base problems
since the submagmas generated by one element are still complicated objects. 
\end{rem}

\section{Instantiations using shifted conjugacy}

\subsection{Protocol 1}
Consider the infinite braid group $(B_{\infty },*)$ with shifted conjugacy as LD-operation. Then the LD-Problem is a
\emph{simultaneous shifted conjugacy problem}.
For $m=n=1$ this becomes the \emph{shifted conjugacy problem} (see e.g. \cite{De06}) which was first solved in \cite{KLT09}
by a double reduction, first to the subgroup conjugacy problem for $B_{n-1}$ in $B_n$, then to an instance of the simultaneous conjugacy problem. 
For the simultaneous conjugacy problem in braid groups we refer to \cite{LL02,KT13}.
If we replace shifted conjugacy by generalized shifted conjugacy, then the corresponding LD-problem still reduces to a subgroup conjugacy problem for
a standard parabolic subgroup of a braid group. Such problems were first solved in a more general framework, namely for Garside subgroups of Garside groups, in \cite{KLT10}. 
Though not explicitly stated in \cite{KLT09, KLT10}, the simultaneous shifted conjugacy problem and its analogue for generalized shifted conjugacy
may be treated by similar methods as in \cite{KLT09, KLT10}.
Though these solutions provide only deterministic algorithms with exponential worst case complexity,
they may still affect the security of Protocol 1 if we use such LD-systems in braid groups as platform LD-systems. 
Moreover, efficient heuristic approaches to the shifted conjugacy problem were developed in \cite{LU08,LU09}.
Therefore, we doubt whether an instantiation of Protocol 1 using shifted conjugacy in braid groups provides a secure KEP.

\subsection{Protocol 2} \label{P2}
Here we propose a natural instantiation of Protocol 2 using generalized shifted conjugacy in braid groups.
Consider the following natural partial multi-LD-system $(B_{\infty }, O_A \cup O_B)$ in braid groups. 
\par 
Let $1<q_1<q_2<p$ such that $q_1, p-q_2 \ge 3$. 
Let any $*_{\alpha }\in O_A$ be of the form
$x*_{\alpha }y= \partial ^p(x^{-1}) \alpha \partial ^p(y)x$ with $\alpha =\alpha _1 \tau _{p,p} \alpha _2$ for some $\alpha _1 \in B_{q_1}$,
$\alpha _2 \in B_{q_2}$.
Analogously, any $*_{\beta }\in O_B$ is of the form $x*_{\beta }y= \partial ^p(x^{-1}) \beta \partial ^p(y)x$ with 
$\beta = \beta _1 \tau _{p,p} \beta _2$ for some $\beta _1  \partial ^{q_2} \in (B_{p-q_2})$, $\beta _2 \in \partial ^{q_1}(B_{p-q_1})$.
Since $[\alpha _1, \beta _1]=[\alpha _1, \beta _2]=[\beta _1, \alpha _2]=1$, the equations (4) and (5) are satisfied. 
Note that, if in addition we have $[\alpha _1, \alpha _2]=[\beta _1, \beta _2]=1$, then 
$(B_{\infty }, *_{\alpha } , *_{\beta })$ is a bi-LD-system according to Proposition \ref{abc} (c).
But in general these additional commutativity relations do not hold for our choice of standard parabolic subgroups as domains for 
$\alpha _1, \alpha _2, \beta _1, \beta _2$. Note that, if we restrict $\alpha _2, \beta _2$ to $\partial ^{q_1}(B_{q_2-q_1})$, 
then these additional relations are enforced.
Anyway, they are not necessary for (4), (5) to hold.
In either case, $\alpha _2$ does not need to commute with $\beta _2$.  
\par  
Then Alice and Bob perform the protocol steps of Protocol 2 for the partial multi-LD-system $(B_{\infty }, O_A \cup O_B)$ 
as described in section \ref{KEPpartial}. 
\par 
The deterministic algorithms from \cite{KLT09, KLT10} do not affect the security of this instantiation of Protocol 2, 
because the \emph{operations are part of the secret}.
More precisely, the LD-endomorphism Search Problem for $S_A$ specifies to the following particular 
simultaneous decomposition problem.
\begin{list}{}{\setlength{\itemsep}{0cm} \setlength{\parsep}{0cm} }
\item[{\sc Input:}]  Element pairs $(s_1,s'_1),\ldots ,(s_m,s'_m)\in B_{\infty }^2$ with 
$$ s'_i=\partial ^p(b^{-1}) \beta _1 \tau _{p,p} \beta _2 \partial ^p(s_i)b$$
for all $i$, $1\le i\le m$, for some (unknown) $b \in B_{\infty }$, $\beta _1 \in \partial ^{q_2}(B_{p-q_2})$, 
$\beta _2 \in \partial ^{q_1}(B_{p-q_1})$.  
\item[{\sc Objective:}] Find $b' \in B_{\infty }$, $\beta '_1 \in \partial ^{q_2} (B_{p-q_2})$, $\beta '_2 \in \partial ^{q_1}(B_{p-q_1})$ such that $$s'_i=\partial ^p((b')^{-1}) \beta '_1 \tau _{p,p} \beta '_2 \partial ^p(s_i)b'$$ for all $i=1,\ldots ,m$.
\end{list}

If we abandon simultaneity, i.e. in the case $m=1$, we obtain a \emph{special decomposition problem}.
In the following section we transform this particular problem to an instance of an apparently new group-theoretic search problem.

\subsection{Conjugacy coset problem}

\begin{defi} \label{SCCP}
Let $H, K$ be subgroups of a group $G$. We call the following problem the subgroup conjugacy coset problem (SCCP) for $(H,K)$ in $G$.
\begin{list}{}{\setlength{\itemsep}{0cm} \setlength{\parsep}{0cm} }
\item[{\sc Input:}]  An element pair $(x,y) \in G^2$ such that $x^G \cap Hy \ne \emptyset $.
\item[{\sc Objective:}]  Find elements $h \in H$ and $c \in K$ such that $cxc^{-1}=hy$.
\end{list}
If $K=G$ then we call this problem the conjugacy coset problem (CCP) for $H$ in $G$.
\end{defi}

This is the search (or witness) version of this problem. The corresponding decision problem is to decide whether
the conjugacy class of $x$ and the left $H$-coset of y intersect, i.e. whether $x^G \cap Hy \stackrel{?}{=} \emptyset $.
Anyway, in our cryptographic context we usually deal with search problems. 
\par 
It is clear from the definition that the SCCP is harder than the double coset problem (DCP) and the subgroup conjugacy problem (subCP), i.e., 
an oracle that solves SCCP for any pair $(H,K)\le G^2$ also solves DCP and subCP.  
\par 
Though the CCP and the SCCP are natural group-theoretic problems, they seem to have attracted little attention in combinatorial group theory so far.
At least we weren't able to find them in the literature.  
\par 
We connect the special decomposition problem from the previous section to the SCCP.

\begin{prop} \label{SDP}
The special decomposition problem (for $m=1$) from section \ref{P2} is equivalent to an instance of SCCP for some standard parabolic subgroups in braid groups,
namely the SCCP for $(\partial ^{q_1}(B_{p-q_1}) \cdot \partial ^{N-p+q_2}(B_{p-q_2}), B_{N-p})$ in $B_N$ for some $N \in \mathbb{N}$.
\end{prop}

\begin{proof}
For $m=1$, we write $s=s_m$ and $s'=s'_m$.
Let $N \in \mathbb{N}$ be sufficiently large such that $s', \partial ^p(s) \in B_N$. For convenience, we choose a minimal $N$ such that $N\ge 2p$.
As in \cite{KLT09} we conclude that $b\in B_{N-p}$ and $\partial ^p(b^{-1}) \in \partial ^p (B_{N-p})$. Therefore we have
$$\tau _{p,N-p}^{-1} \partial ^p(b^{-1})=b^{-1}\tau _{p,N-p}^{-1}. $$
Furthermore, since $\tau _{p,p}\beta _2=\partial ^p(\beta _2) \tau _{p,p}$ and
$\tau _{p,N-p}^{-1}\tau _{p,p}=\partial ^p (\tau _{p,N-2p}^{-1})$ for $N\ge 2p$\footnote{If $N<2p$ then $\tau _{p,N-p}^{-1}\tau _{p,p}=\partial ^p (\tau _{p,2p-N})$. But for generic instances $N$ is expected to be much larger than $2p$.}, 
we get
\[ \begin{array}{rcll}
s' &=& \partial ^p(b^{-1}) \beta _1 \tau _{p,p} \beta _2 \partial ^p(s)b  & \Leftrightarrow \\
\tau _{p,N-p}^{-1}s' &=& b^{-1} \tau _{p,N-p}^{-1} \beta _1 \partial ^p(\beta _2) \tau _{p,p} \partial ^p(s)b & \\
 &=&  b^{-1} \partial ^{N-p} (\beta _1)\beta _2 \tau _{p,N-p}^{-1}\tau _{p,p} \partial ^p(s)b & \Leftrightarrow \\
 b\tilde{s}' b^{-1} &=& \tilde{\beta } \cdot \tilde{s} &
\end{array} \]
with $\tilde{s}'=\tau _{p,N-p}^{-1}s'$, $\tilde{s}=\partial ^p (\tau _{p,N-2p}^{-1} s)$, and 
$$\tilde{\beta }=\partial ^{N-p} (\beta _1)\beta _2 
\in \partial ^{q_1}(B_{p-q_1}) \cdot \partial ^{N-p+q_2}(B_{p-q_2}).$$
\end{proof}

\par Recall that the algorithms from \cite{KLT09, KLT10}, as well as from \cite{GKLT13}, only solve
instances of the subgroup conjugacy problem for parabolic subgroups of braid groups, partially by transformation to the simultaneous conjugacy problem
in braid groups \cite{KTV13}.
No deterministic or even heuristic solution to the SCCP for (standard) parabolic subgroups in braid groups is known yet. 

\medskip

{\bf Open problem.} Find a solution to the SCCP, or even the CCP, for (standard) parabolic subgroups in the braid group $B_N$. 

\medskip

The CCP (and the SCCP) appear to be inherently quad\-ratic, i.e. we do not see how it may be linearized such that linear algebra attacks
as the \emph{linear centralizer attack} of B. Tsaban \cite{Ts12} apply. It shares this feature with Y. Kurt's {\it Triple Decomposition Problem} 
(see section 4.2.5. in \cite{MSU11}).  

\begin{rem}
The reader might be slightly disappointed that we did not offer a proposal with concerete parameter values to get excited about.
The reasons are twofold. First, the purpose of this article is to provide a general scheme how to get from any LD-, multi-LD-, or even other 
left distributive system a key establishment protocol. In this sense we provide a variety of KEP instantiations. 
\par 
Second, and this is the main reason, our proposal is not the end of the story. Indeed, in an upcoming paper \cite{KT13}
we suggest further improved KEPs for all LD- and multi-LD-systems etc., namely, systems based on iterated versions of the LD-problem.
There we will provide more concrete proposals, even efficient instantiations in finite groups which we do not consider as secure platforms for
Protocols 1 and 2.  
Nevertheless, the reader might feel free to attack, for example, the instantiation of Protocol 2 using generalized shifted conjugacy in braid groups for some small parameter values which still resist brute force attack. 
\par 
Further ideas and open problems for instantiating Protocol 1 and 2 are contained in \cite{Ka12}. 
\end{rem}

{\bf Acknowledgements.} The first author acknowledges financial support by the Minerva Foundation of Germany. 
Parts of the paper were written up during the stay of the first author at University of Queensland, Brisbane.
For this stay the first author acknowledges financial support by the Australian Research Council (project DP110101104). 
\par 
The second author acknowledges financial support by The Oswald Veblen Fund. 
We thank Boaz Tsaban for encouragement and fruitful discussions.
We are grateful to Ciaran Mullan for many useful comments.

\end{document}